\documentclass[letterpaper, 10 pt, journal, twoside]{ieeetran}
\ifCLASSINFOpdf

 \usepackage{hyphenat}
\usepackage{xcolor}
\usepackage{amsmath}
\usepackage{tabularx}
\usepackage{graphicx,cite}
\usepackage{amssymb}
\usepackage{amsthm}
\usepackage{algorithm,algpseudocode}
\usepackage{mathabx}
\newtheorem{theorem}{Theorem}
\newtheorem{corollary}{Corollary}
\newtheorem{lemma}{Lemma}
\newtheorem{definition}{Definition}
\newtheorem{remark}{Remark}
\newtheorem{assumption}{Assumption}
\newtheorem{proposition}{Proposition}
\usepackage{hyperref}
\usepackage{subcaption}
\title{\LARGE \bf
Locally-Aware Constrained Games on Networks
}

\author{Guanze Peng,  Tao Li, Shutian Liu, Juntao Chen, and Quanyan Zhu
\thanks{G. Peng, T. Li, S. Liu and Q. Zhu are with the Department of Electrical and Computer Engineering, New York University, Brooklyn, NY, 11201,  USA. Email:
        {\tt\small {\{gp1363, sl6803, tl2636, qz494\}}@nyu.edu}}%
\thanks{J. Chen is with the Department of Computer and Information Sciences,
        Fordham University, New York, NY, 10023, USA.
        Email: {\tt\small jchen504@fordham.edu}}%
}

\begin{document}

\maketitle
\thispagestyle{empty}
\pagestyle{empty}

\begin{abstract}
Network games have been instrumental in understanding strategic behaviors over networks for applications such as critical infrastructure networks, social networks, and cyber-physical systems. One critical challenge of network games is that the behaviors of the players are constrained by the underlying physical laws or safety rules, and the players may not have complete knowledge of network-wide constraints.
To this end, this paper proposes a  game framework to study constrained games on networks, where the players are locally aware of the constraints. We use \textit{awareness levels} to capture the scope of the network constraints that players are aware of. We first define and show the existence of generalized Nash equilibria (GNE) of the game, and point out that higher awareness levels of the players would lead to a larger set of GNE solutions. We use necessary and sufficient conditions to characterize the GNE, and propose the concept of the dual game to show that one can convert a locally-aware constrained game into a two-layer unconstrained game problem. We use linear quadratic games as case studies to corroborate the analytical results, and in particular, show the duality between Bertrand games and Cournot games.

\end{abstract}

\section{Introduction}
  Research on networks has been received a significant amount of attention over the last few decades for their applications in many essential fields, such as infrastructure systems, wireless communications, and the Internet of Things (IoT) \cite{kurian2017electric,manshaei2013game,chen2019optimal}. Game theory provides a bottom-up approach to analyze and design networks, especially when the network is large. The equilibrium solution concepts provide a quantitative understanding of the outcome of strategic behaviors over networks. They have become the standard tools in the design of communication networks \cite{alpcan2002cdma}, transportation networks \cite{altman2004equilibrium}, and energy systems \cite{chen2016game}.  
There are several critical challenges in applying game theory to networks in engineering.  First, the games played on networks are usually subject to  constraints that arise from physics, such as Kirchhoff's laws for power systems and conservation laws in water systems, or safety, such as security compliance and rules. Hence, the equilibrium of the network games needs first and foremost to satisfy these network constraints. Second, not every player in the network is aware of the same constraints over the network. In most cases, players are only locally aware of a subset of constraints and then make decisions. This feature arises from the fact that players have bounded rationality. Being fully aware of the network constraints is not always possible, especially for large networks. 

In this work, we propose a locally-aware constrained non-cooperative game on networks to address these challenges. The players have only local knowledge of the constraints at his node and his neighbors. We propose the concept of \textit{awareness} to capture the knowledge of the subset of constraints that the players need to satisfy. We fully characterize the sufficient and necessary conditions of the equilibrium solution under diverse awareness of constraint sets. We show that the proposed framework generalizes the extreme cases where  all the players have the same and complete knowledge of the network constraints. We transform the locally-aware games to an equivalent dual game, which consists of two hierarchical unconstrained games. As a case study, we use linear quadratic games to numerically corroborate the analytical results. In particular, we show that Cournot games and Bertrand games can be viewed as the dual games of each other.   {
The proposed game is applicable to various applications. For example, in the distributed generation systems, each player determines the amount of power to generate by considering the demand and capacity constraints. It is possible that the players may have heterogeneous level of awareness on these constraints, yielding different equilibrium strategies. We further corroborate the results with an application to the energy system.}

Our work is closely related to the literature on constrained games, e.g., \cite{pavel2007extension,charnes1953constrained}, where the set of constraints is common knowledge and the players have the same set of constraints.
The equilibrium concept that we investigate in this work is the {generalized Nash equilibrium} (GNE), which has been studied in \cite{facchinei2009generalized,debreu1952social}. The computational methods of GNEs have been studied in \cite{paccagnan2016distributed,facchinei2011computation}.   {This paper falls in the category of constrained games with asymmetric constraint information, which have been investigated in \cite{pavel2019distributed}, \cite{yi2019operator}, \cite{yi2019asynchronous}, \cite{belgioioso2019distributed}. Most of the existing work focus on the algorithm design to find GNEs of this class of games, while this work aims to characterize the GNE and study the impact of the information structure on the GNEs.
}

The rest of the paper is organized as the following: Section~\ref{PROBLEM STATEMENT} formally defines the locally-aware constrained games.
Section~\ref{GLOBALLY-SHARED CONSTRAINED GAME} revisits the results in the globally-aware constrained game. In Section~\ref{LOCALLY-SHARED CONSTRAINED GAME}, we focus on the locally-aware constrained game and its relationship with globally-aware constrained game. Moreover, we show that under certain assumptions, a locally-aware constrained game can be decomposed into two unconstrained games. In Section~\ref{Application}, we apply the decomposition result to enunciate the relationship between Bertrand game and Cournot game. We use two linear-quadratic game examples to corroborate the theoretic results. Section~\ref{CONCLUSIONS AND FUTURE WORK} concludes the paper and points out the possible direction of the future work.

The rest of the paper is organized as the following: Section~\ref{PROBLEM STATEMENT} formally defines the locally-aware constrained games.
Section~\ref{GLOBALLY-SHARED CONSTRAINED GAME} revisits the results in the globally-aware constrained game. In Section~\ref{LOCALLY-SHARED CONSTRAINED GAME}, we focus on the locally-aware constrained game and its relationship with globally-aware constrained game. Moreover, we show that under certain assumptions, a locally-aware constrained game can be decomposed into two unconstrained games. In Section~\ref{Application}, we apply the decomposition result to enunciate the relationship between Bertrand game and Cournot game. We use two linear-quadratic game examples to corroborate the theoretic results. Section~\ref{CONCLUSIONS AND FUTURE WORK} concludes the paper and points out the possible direction of the future work.

\section{Problem Statement}\label{PROBLEM STATEMENT}

\subsection{Game Framework}
Consider a connected network, which contains $N$ nodes (players) described by an undirected graph $\mathcal{G}=\left(\mathcal{N},\mathcal{E}\right)$ with $\mathcal{N}=\{1,2,...,N\}$ as the index set of the nodes. The set of edges is denoted by $\mathcal{E}$ which contains the links between connected nodes. For each node $i\in\mathcal{N}$, denote the set of its neighbors by $\mathcal{N}_i$ (including himself) and the number of its neighbors by $|\mathcal{N}_i|$. Suppose that a noncooperative constrained game is played among the $N$ players. For each player $i\in\mathcal{N}$, let $\mathcal{U}_i\subseteq\mathbb{R}^{K_i}$ be the continuous compact action set, where $K_i$ is a positive integer, and let $f_i:\prod_{i=1}^N\mathcal{U}_i\rightarrow \mathbb{R}$ be the cost function of Player $i$. Denote Player $i$'s action by $u_i$, and let $u=(u_j)_{j\in\mathcal{N}}\in\mathcal{U}:=\prod_{{j\in\mathcal{N} }}\mathcal{U}_j$, and $u_{-i}=(u_j)_{j\in\mathcal{N} \setminus\{i\}}\in\mathcal{U}_{-i}:=\prod_{{j\in\mathcal{N} \setminus\{i\}}}\mathcal{U}_j$. The constrained game can be defined as a tuple,
\begin{equation}
    \Xi:=( \mathcal{N},\{\mathcal{N}_i,f_i,\mathcal{U}_i,\mathcal{F}_i\}_{i\in\mathcal{N}},\{\hat{\Omega}_m\}_{m\in\mathcal{M}},\mathcal{A},\mathfrak{s}),
\end{equation}
where 
\begin{enumerate}
        \item $\hat{\Omega}_m\subseteq\mathbb{R}^{\prod_{i=1}^NK_i},m\in\mathcal{M}=\{1,2,...,M\},$ is a feasible set that represents a constraint on the tuple of all the players' actions, $u$;
        \item $\mathcal{A}$ is a finite set of \textit{awareness} levels;
        \item $\mathcal{F}_i$ is a subset of $\Omega_i$, i.e., $\mathcal{F}_i\subseteq \Omega:=\{\hat{\Omega}_m\}_{m\in\mathcal{M}}$, $i\in\mathcal{N}$;
        \item $\mathfrak{s}:\mathcal{A}\rightarrow 2^\Omega$, is a mapping associating an awareness level with a subset of constraints, where $2^\Omega$ is the power set of $\Omega$.
\end{enumerate}

We call an element of $\mathcal{A}$ an \textit{awareness level}; $\Omega_i$ the \textit{awareness set} of Player $i$, which quantifies Player $i$'s knowledge of the constraints. At the beginning of the game $\Xi$, nature \`a la Harsarnyi \cite{harsanyi1967games} determines the awareness level of each player from $\mathcal{A}$. His awareness level indicates his knowledge of a feasible subset of $\{\hat{\Omega}_m\}_{m\in\mathcal{M}}$ via the mapping $\mathfrak{s}$. Specifically, let the awareness set be $\mathcal{A}=2^{\mathcal{M}}.$
Denote a subset of $\mathcal{M}$ by $\tilde{\mathcal{M}}\in 2^{\mathcal{M}}$ (i.e., $\tilde{\mathcal{M}}\subseteq{\mathcal{M}}$). Then, $\mathfrak{s}({\tilde{\mathcal{M}}})=\{\hat{\Omega}_{m}\}_{m\in\tilde{\mathcal{M}}}.$
In other words, the players of awareness level ${\tilde{\mathcal{M}}}$ are aware of $\{\hat{\Omega}_{m}\}_{m\in\tilde{\mathcal{M}}}$.
\begin{definition}\label{moreaware}
In the constrained game $\Xi$, if two awareness levels ${\tilde{\mathcal{M}}},{\hat{\mathcal{M}}}\in\mathcal{A}$, satisfy 
$ \mathfrak{s}({\tilde{\mathcal{M}}})\subseteq\mathfrak{s}({\hat{\mathcal{M}}}),$
    {or equivalently,} 
    ${\tilde{\mathcal{M}}}\subseteq{\hat{\mathcal{M}}},$
then we say that the player with awareness level ${\hat{\mathcal{M}}}$ is \textit{more aware} than the player with awareness level ${\tilde{\mathcal{M}}}$.
\end{definition}
By Definition~\ref{moreaware}, the awareness level set $\mathcal{A}$ can be considered as a Boolean lattice, which is a partially ordered set. This fact enables us to compare the awareness levels. In a locally-aware constrained game, the players can observe the constraints of his neighbors. Denote the awareness level of Player $i$ as ${\tilde{\mathcal{M}}_i}$, $i\in\mathcal{N}$. Then, according to the description of the game, the locally-aware constraint set of Player $i$ is 
\begin{equation}\label{infoset}
\begin{aligned}
\mathcal{F}_i=\mathfrak{s}\left({\bigcup_{j\in\mathcal{N}_i}\tilde{\mathcal{M}}_{j}}\right)=\left\{\hat{\Omega}_{m}\right\}_{m\in\tilde{\mathcal{M}}_{j},{j}\in\mathcal{N}_i}.
\end{aligned}
\end{equation}
We can see that acquiring the knowledge of the neighbors' constraints is equivalent to increasing the awareness level of the players, as for each player $i\in\mathcal{N}$ and every possible $\mathcal{N}_i$, we  have 
$\mathfrak{s}({\tilde{\mathcal{M}}_{i}})\subseteq\mathfrak{s}({\bigcup_{j\in\mathcal{N}_i}\tilde{\mathcal{M}}_{j}}).$

Note that the locally-aware games framework considers a different aspect of the incomplete-information games. Generally speaking, the knowledge in games can be classified into three categories: (i) the knowledge that a player knows that he knows, (ii) the knowledge that a player knows that he does not know, and (iii) the knowledge that player does not know that he does not know. There is a rich literature on the information structure induced by the first two classes of knowledge (see \cite{harsanyi1967games,aumann1995repeated}), whilst the awareness framework investigates the third class of knowledge. In the locally-aware games framework, the player is only aware of the constraints that he knows, but unaware of the existence of other constraints.

\subsection{Equilibrium Concept and Existence}

From the information set given in $\eqref{infoset}$, each player $i\in\mathcal{N}$ faces the parameterized optimization problem given by the following:
\begin{equation}\label{unsharedoptori}
    \begin{aligned}
  \min_{u_i\in\mathcal{U}_i}&\quad f_i(u_i,u_{-i})\\
        \text{s.t.}&\quad  (u_i,u_{-i})\in\Omega_{j},\quad {j}\in\mathcal{N}_i,
    \end{aligned}
\end{equation}
where $\Omega_{j}:=\bigcap_{m\in\tilde{\mathcal{M}}_{j}}\hat{\Omega}_{m}$. We assume that the feasible set can be represented algebraically by functions ${g}_i:\mathcal{U}\rightarrow \mathbb{R}$, $i\in\mathcal{N}$, as follows:
\begin{equation*}
    \Omega_i=\left\{(u_i,u_{-i}):{g}_{j}(u_i,u_{-i})\leq {0}, j\in\mathcal{N}_i\right\}.
\end{equation*}
Therefore, \eqref{unsharedoptori} can be rewritten as
\begin{equation}\label{unsharedopt}
    \begin{aligned}
     \textbf{OP}_i:\quad   \min_{u_i\in\mathcal{U}_i}&\quad f_i(u_i,u_{-i})\\
        \text{s.t.}&\quad  {g}_{j}(u_i,u_{-i})\leq {0},\quad j\in\mathcal{N}_i.
    \end{aligned}
\end{equation}
Given $u_{-i}$, the best response of Player $i$ is the optimal solution to \eqref{unsharedopt} parameterized by $u_{-i}$. We denote the best response of Player $i$ by $BR_i(u_{-i})$, which is a mapping $BR_i: \mathcal{U}_{-i}\rightarrow\mathcal{U}_i$.


A suitable equilibrium concept for the constrained games is generalized Nash equilibrium (GNE) given by the following.

\begin{definition}\label{usGNE}
A tuple $u^*=(u^*_i,u_{-i}^*)$ is said to be a GNE solution for the game $\Xi$, if $
    f_i(u^*_i,u_{-i}^*)\leq f_i(u_i,u_{-i}^*),\  \forall u_i\in \Omega_i(u_{-i}^*),\  i\in\mathcal{N},$
where $\Omega_i(u_{-i}^*):=\left\{u\in\mathcal{U}_i:{g}_{j}(u,u^*_{-i})\leq {0},\ j\in\mathcal{N}_i\right\}.$
\end{definition}

Define the global feasible set by $
    \tilde{\Omega}:=\left\{u:{g}_i(u)\leq {0},\ \forall\ i\in\mathcal{N}\right\}$,
which consists of all the feasible solutions that satisfy all the constraint. We make the following regularity assumptions.
\begin{assumption}[Regularity Condition]\label{bigassump}
The following conditions hold throughout this work:
\begin{enumerate}
    \item[a)] $f_i$ and ${g}_i$, $i\in\mathcal{N},$ are continuously differentiable;
    \item[c)] $\mathcal{U}_i$, $i\in\mathcal{N},$ are convex and compact;
    \item[d)] $\tilde{\Omega}$ is nonempty;
    \item[e)] For all $i\in\mathcal{N}$, given any ${u}_{-i}\in\mathcal{U}_{-i}$, \textbf{OP}$_i$ is a convex optimization problem.
\end{enumerate} 
\end{assumption}

\begin{assumption}[Collective Awareness Condition]
The following assumption regarding awareness is assumed to hold:
\begin{equation*}
    \bigcup_{i\in\mathcal{N}}\tilde{\mathcal{M}}_i=\mathcal{M}.
\end{equation*}
\end{assumption}

\begin{theorem}\label{thm1}
 Let $BR_i(u_{-i})$ be the optimal solution to \textbf{OP}$_i$, $i\in\mathcal{N}$. The following fixed-point equation admits at least one fixed-point solution:
 \begin{equation}\label{assumpfeq}
     u\in BR(u):=\begin{bmatrix}
     BR_1(u_{-1})\\
     BR_2(u_{-2})\\
     \vdots\\
     BR_N(u_{-N})
     \end{bmatrix}
     ,
 \end{equation}
where $BR:\mathcal{U}\rightarrow 2^\mathcal{U}$. Moreover, there exists at least one GNE in $\Xi$.
\end{theorem}
\begin{proof}
The proof directly follows from \cite{debreu1952social}, which uses the fixed-point theorem developed in \cite{eilenberg1946fixed}.
\end{proof}





\section{Globally-Aware Constrained Game}\label{GLOBALLY-SHARED CONSTRAINED GAME}
This section studies a special extreme case of the globally-aware constrained games, where the players are aware of all the constraints in the game.
\begin{definition}
In the constrained game $\Xi$, a player with awareness level ${\hat{\mathcal{M}}}$ is said to be globally aware if 
$       \mathfrak{s}({\tilde{\mathcal{M}}})\subseteq\mathfrak{s}({\hat{\mathcal{M}}}),\ \forall\ \tilde{\mathcal{M}}\in2^{\mathcal{M}}.
$
\end{definition}
By definition, if the awareness level of a player is $\mathcal{M}$, then he is globally aware. In a globally-aware constrained game, all the players are globally aware, i.e., each player's awareness level is ${\mathcal{M}}$. The tuple that represents the globally-aware constrained games is given by
\begin{equation}
  \tilde{\Xi}:=( \mathcal{N},\{\mathcal{N}_i,f_i,\mathcal{U}_i,\tilde{\mathcal{F}}_i\}_{i\in\mathcal{N}},\{\hat{\Omega}_m\}_{m\in\mathcal{M}},\mathcal{A},\mathfrak{s}).
\end{equation}
The awareness set in $\tilde{\Xi}$ for each player  $i\in\mathcal{N}$ is
$ \tilde{\mathcal{F}}_i=\{\hat{\Omega}_m\}_{m\in\mathcal{M}}.$

The following corollary directly follows from Theorem \ref{thm1}.

\begin{corollary}\label{existinshared} 
There exists at least one GNE in $\tilde{\Xi}$.
\end{corollary}

\begin{remark}\label{globalawareness1}
We observe that a globally-aware constrained game played on arbitrary network can be considered as a locally-aware constrained game played on a  fully-connected network, where a link exists between two arbitrary players. 
\end{remark}

In the classical constrained game, for example in \cite{pavel2007extension}, the awareness set of Player $i$ is given by the following:
$
    \bar{\mathcal{F}}_i=\left\{\bigcap_{m\in\mathcal{M}}\hat{\Omega}_{m}\right\},\ i\in\mathcal{N},
$
under which each player knows all the constraints. However, he is not aware of the fact that he knows all the constraints.
Then, by definition, we can show that the GNEs in the game defined via $\bar{\Xi}:=( \mathcal{N},\{\mathcal{N}_i,f_i,\mathcal{U}_i,\}_{i\in\mathcal{N}},\{\bigcap_{m\in\mathcal{M}}\hat{\Omega}_m\})$, and $\tilde{\Xi}$ coincide. This observation indicates that the classical constrained game is equivalent to a globally-aware constrained game.

\begin{remark}\label{globalawareness}
Based on Remark~\ref{globalawareness1}, a globally-aware constrained game is equivalent to a locally-aware constrained game played on a fully-connected network. Thus, the locally-aware constrained games can be seen as an extension of the globally-aware constrained game.
\end{remark}

\section{Locally-Aware Constrained Game}\label{LOCALLY-SHARED CONSTRAINED GAME}

In this section, we study the properties of the locally-aware constrained games.

\subsection{Connectivity of Networks, Awareness and GNEs}
Consider an undirected connected network, $\underline{\mathcal{G}}=\left\{\mathcal{N},\underline{\mathcal{E}}\right\}$. If $\mathcal{E}\subseteq\underline{\mathcal{E}}$, we say that $\underline{\mathcal{G}}$ has a stronger connectivity than $\mathcal{G}$, as $\underline{\mathcal{G}}$ contains more links. 
\begin{theorem}\label{awarenessandconnectivity}
As the connectivity of the network on which a locally-aware constrained game is played increases, the set of GNEs of that game becomes larger.
\end{theorem}
\begin{proof}
Consider a locally-aware constrained game on network $\mathcal{G}$ and $\underline{\mathcal{G}}$, where  $\underline{\mathcal{G}}$ has a higher connectivity than $\mathcal{G}$. Denote the sets of neighbors on $\underline{\mathcal{G}}$ by $\left\{\underline{\mathcal{N}}_i\right\}_{i\in\mathcal{N}}$. As there are more links connecting the players, $\mathcal{N}_i\ \subseteq\ \underline{\mathcal{N}}_i,\  i\in\mathcal{N}.$
Similar to \eqref{infoset}, it is clear that $\mathcal{F}_{i}\ \subseteq\ \underline{\mathcal{F}}_{i},\ i\in\mathcal{N}.$
Thus, on the network $\underline{\mathcal{G}}$,  each player $i\in\mathcal{N}$ seeks $u_i\in \mathcal{U}_i$ to solve the following optimization problem:
\begin{equation}\label{awarenessresult}
    \begin{aligned}
     \underline{\textbf{OP}}_i:\quad   \min_{u_i\in\mathcal{U}_i}&\quad f_i(u_i,u_{-i})\\
        \text{s.t.}&\quad  (u_i,u_{-i})\in\Omega_{j},\quad {j}\in\underline{\mathcal{N}}_i.
    \end{aligned}
\end{equation}
It is clear that $ \underline{\Omega}_i(u_{-i}):=\bigcap_{j\in\underline{\mathcal{N}}_i}\Omega_{j}(u_{-i})\subseteq{\Omega}_i(u_{-i})=\bigcap_{j\in\mathcal{N}_i}\Omega_{j}(u_{-i})$, $ {i\in\mathcal{N}}$.
We proceed the rest of the proof by contradiction. Suppose that there exists ${u}^*$ which is a GNE solution to $\Xi$, but not a GNE solution to $\underline{\Xi}$ (here, $\Xi$ and $\underline{\Xi}$ are the game played on networks $\mathcal{G}$ and $\underline{\mathcal{G}}$, respectively). Without loss of generality, suppose that there exists some $\tilde{u}_{i}^*\in \underline{\Omega}_i(u_{-i}^*).$ such that
$
    \begin{aligned}
    f_i(u^*_i,{u}_{-i}^*)> f_i(\tilde{u}_{i}^*,u_{-i}^*).
    \end{aligned}
$
However, as $\underline{\Omega}_i(u_{-i}^*)\subseteq \Omega_i(u_{-i}^*)$,
$
    f_i(u_i^*,u_{-i}^*)\leq f_i(u_i,u_{-i}^*),\ \forall u_i\in\underline{\Omega}_i(u_{-i}^*).
$
By letting $u_i=\tilde{u}_i^*$, we observe that two inequalities contradict each other.
\end{proof}

\begin{remark}\label{remarkconnectivity}
The higher connectivity of the network enhances the player's awareness, as
$
    \mathfrak{s}({\bigcup_{j\in\mathcal{N}_i}\tilde{\mathcal{M}}_{j}})\subseteq\mathfrak{s}({\bigcup_{j\in\underline{\mathcal{N}}_i}\tilde{\mathcal{M}}_{j}}),\  i\in\mathcal{N}.
$
Theorem~\ref{awarenessandconnectivity} states that, as players have more awareness about the game, there could be more equilibrium points as a result.
\end{remark}

\subsection{GNEs in $\Xi$ and GNEs in $\tilde{\Xi}$}
The following corollary presents the relationship between GNEs in $\Xi$ and GNEs in $\tilde{\Xi}$.
\begin{corollary}\label{gneandgne}
If ${u}^*=(u_i^*,{u}_{-i}^*)$ is a GNE solution to $\Xi$, then it is also a GNE solution to $\tilde{\Xi}$.
\end{corollary}

\begin{remark}
As mentioned above, a globally-aware constrained game can be considered as a locally-aware constrained game played on a fully-connected network. Hence, the fully-connected network structure enhances the players to the highest possible awareness level (i.e. the global awareness) in $\mathcal{A}$. From this perspective, according to Remark~\ref{remarkconnectivity}, Corollary~\ref{gneandgne} is an immediate result of Theorem~\ref{awarenessandconnectivity}.
\end{remark}

Based on Corollary \ref{gneandgne}, we further obtain the following result.
\begin{corollary}
If there exists a unique GNE solution to $\tilde{\Xi}$, then there also exists a unique GNE solution to ${\Xi}$. Moreover, the GNE solution to $\tilde{\Xi}$ coincides with the GNE solution to the corresponding ${\Xi}$. 
\end{corollary}

\subsection{GNEs in $\Xi$ and VEs in $\tilde{\Xi}$}
In the constrained games, the GNEs in $\tilde{\Xi}$ are not of practical interest in the general cases and have less economic justifications as the concept of GNE can be far too weak to serve as a solution concept. Therefore, the variational equilibrium (VE) is proposed as a refinement of GNE \cite{kulkarni2012variational}. 
\begin{definition}
A tuple $u^*=(u_i^*,u_{-i}^*)$ is said to be a VE of $\tilde{\Xi}$ if $u^*$ satisfies $
   F^{\rm T}(u^*)\cdot(u - u^*)\geq 0,\ \forall\ u\in\tilde{\Omega},$
where $F(u)=[\nabla_{u_1}f^{\rm T}_1(u_1,u_{-1}),\cdots,\nabla_{u_N}f^{\rm T}_N(u_N,u_{-N})]^{\rm T}$.
\end{definition}
\begin{proposition}[\cite{kulkarni2012variational}]
 Under Assumption~\ref{bigassump}, a VE of $\tilde{\Xi}$ is also a GNE of $\tilde{\Xi}$.
\end{proposition}    
A GNE in $\Xi$ is not necessary a VE in $\tilde{\Xi}$. Consider the following counterexample in a two-player game. Suppose $f_1(u_1,u_2)=\frac{1}{2}u_1^2-u_1$, $f_2(u_1,u_2)=\frac{1}{2}u_2^2-2u_2$, and $\Omega_1=\left\{(u_1,u_2):u_1+u_2\leq 1\right\}$.
The GNE solution to $\Xi$ is $(-1,2)$, and the VE solution to $\tilde{\Xi}$ is $(0,1)$.

\subsection{Characterization of GNEs in $\Xi$}
We characterize GNEs in $\Xi$ and provide the necessary and sufficient conditions.

\begin{lemma}\label{suffice}
Let ${u}^*$ be a feasible solution such that
\begin{equation}\label{auguopt}
    \tilde{f}({u}^*;{u}^*)\leq \tilde{f}(v;{u}^*),\ \ \forall\ v\in \Omega({u}^*), 
\end{equation}
where $
     \tilde{f}(v;u)=\sum_{i\in\mathcal{N}}f_i(v_i,u_{-i}),
$ and $
    \Omega(u)=\prod_{i\in\mathcal{N}}\Omega_i(u_{-i}).
$
Then, ${u}^*$ is a GNE solution to $\Xi$.
\end{lemma}

\begin{proof}
We proceed the proof by contradiction. Suppose that ${u}^*$ satisfies \eqref{auguopt} and ${u}^*$ is not a GNE solution to $\Xi$.  Without loss of generality, suppose that there exists some $\Bar{u}_i\in\Omega_i({u}^*_{-i})$ such that
\begin{equation*}
    \tilde{f}_i({u}_i^*,{u}^*_{-i})> \tilde{f}_i(\Bar{u}_i,{u}^*_{-i}).
\end{equation*}
By adding both sides by $\sum_{j\in\mathcal{N}/\{i\}}f_j({u}_j^*,{u}^*_{-j})$, we have
\begin{equation*}
\begin{aligned}
    \tilde{f}({u}^*;{u}^*)&=f_i({u}_i^*,{u}^*_{-i})+\sum_{j\in\mathcal{N}/\{i\}}f_j({u}_j^*,{u}^*_{-j})\\
    &> f_i(\bar{u}_i,{u}^*_{-i})+\sum_{j\in\mathcal{N}/\{i\}}f_j({u}_j^*,{u}^*_{-j})\\
    &={\tilde{f}}(\Bar{{u}};{u}^*),
\end{aligned}
\end{equation*}
where $\Bar{{u}}=(\Bar{u}_i,{{u}}^*_{-i}).$ This contradicts our assumption that ${u}^*$ satisfies \eqref{auguopt}. 
\end{proof}

\begin{remark}
Note that the conditions in Lemma~\ref{suffice} also suffice to show that any $u^*$ satisfying \eqref{auguopt} is a GNE in $\tilde{\Xi}$ as shown in \cite{pavel2007extension}.
\end{remark}

Define the Lagrangian functions
\begin{equation}\label{lagrange}
    \mathcal{L}_i\left(u_i;u_{-i};{\mu}_i\right):=f_i(u_i,{u}_{-i})+\sum_{j\in\mathcal{N}_i}{\lambda}_{j}{g}_{j}(u_i,{u}_{-i}),
\end{equation}
and the augmented Lagrangian function
\begin{equation}\label{sperate}
    \mathcal{L}(v;{u};{\mu}):=\sum_{i\in\mathcal{N}}\mathcal{L}_i(v_{i};{u}_{-i};{\mu}),
\end{equation}
where $\mu_i:=(\lambda_{j})_{j\in\mathcal{N}_i}$ and $\mu=(\mu_1,\mu_2,...,\mu_N)$.
\begin{remark}
Note that the Lagrangian function \eqref{lagrange} captures the local information of each player $i$. The Lagrangian multiplier, $\mu_i$, $i\in\mathcal{N}$, stands for the awareness level to some extent. We call $\mu_i$ the awareness variable of Player $i$, and the dimension of $\mu_i$ implies the awareness level of Player $i$.
\end{remark}

\begin{theorem}[Necessary and Sufficient Condition for GNEs]\label{neceandsuffic}
 The tuple $u^*=(u^*_i,u_{-i}^*)$ is a GNE in $\Xi$ if and only if 
        \begin{equation}\label{mostimportant}
        u^*=\arg\min_{v\in\mathcal{U} }\left\{\max_{{{\mu}\geq 0}}\mathcal{L}(v;u,{{\mu}})\right\}\bigg|_{v=u}.
    \end{equation}
\end{theorem}
\begin{proof}
We first show the necessity. If $u^*$ is a GNE solution, then $u^*$ is feasible, i.e., $u_i^*\in\Omega_i(u^*_{-i})$, $\forall i\in\mathcal{N}$. By the definition of GNE, it is equivalent to solving the following $N$ intertwined optimizations: with fixed $u=(u_i,u_{-i})$, for each $i\in\mathcal{N}$,
    \begin{equation}\label{indi}
        \begin{aligned}
        \gamma_i^*(u)=\arg\min_{v_i\in\mathcal{U}_i}\max_{{{\mu}}_i\geq 0}\ f_i(v_i,u_{-i})+{{\mu}}_i{g}_i(v_i,u_{-i}).
        \end{aligned}
    \end{equation}
    As we assume that \textbf{OP}$_i$ is a convex optimization problem, there always exists a solution to \eqref{indi}. To obtain GNE, we need to solve the following $N$ fixed-point equations 
        $\gamma_i^*(u)\ =\ u_i,\  i\in\mathcal{N}.$
    Thus, by collecting $N$ fixed-point equations, we can write down the necessary condition equivalently and more compactly as
    \begin{equation*}
        u^*=\arg\min_{v\in\mathcal{U} }\left\{\max_{\mu\geq 0}{\mathcal{L}}(v;u,{{\mu}})\right\}\bigg|_{v=u}.
    \end{equation*}
   We next show the sufficiency. By the definition of $\mathcal{L}$ in \eqref{sperate},
    we can rewrite the inner maximization in the sufficient condition \eqref{mostimportant} as
    \begin{equation}\label{innermax}
    \begin{aligned}
        \max_{{\mu}\geq 0}\mathcal{L}(v;{u};{\mu})&= \max_{{\mu}\geq 0}\sum_{i\in\mathcal{N}}\mathcal{L}_i(v_{i},{u}_{-i},{\mu}_i)\\
        &=\sum_{i\in\mathcal{N}}\max_{{\mu}_i\geq 0}\mathcal{L}_i(v_{i},{u}_{-i},{\mu}_i).
    \end{aligned}
    \end{equation}
    Denote the maximizer of \eqref{innermax} by $\zeta_i^*(v_i,{u}_{-i})$, $i\in\mathcal{N}$. In order to minimize $\mathcal{L}_i(v_i,{u}_{-i},\zeta_i^*(v_i,{u}_{-i}))$, $v_i$ must be chosen such that ${g}_i(v_i,{u}_{-i})\leq {0}$. Otherwise, if ${g}_i(v_i,{u}_{-i})> {0}$ then by choosing a functional such that $\zeta_i^*(v_i,{u}_{-i})>0$, the minimax value, $\min_{v_i\in\mathcal{U}_i}\max_{{\mu}_i\geq 0}\mathcal{L}_i(v_{i},{u}_{-i},{\mu}_i)$, will go to infinity. Then, we solve for the fixed-point equation
    \begin{equation*}
        {u}^*=\arg\min_{{v}\in\mathcal{U} }\left\{\sum_{i\in\mathcal{N}}\mathcal{L}_i(v_i;{u}_{-i};\zeta_i^*(v_i,u_{-i})\right\}\bigg|_{v=u}.
    \end{equation*}
    Therefore, $\zeta^{*\rm T}_i(u^*)\mathbf{g}_i(u^*)=0.$ Moreover, since ${u}^*$ optimizes $L(v;u;\left\{\zeta^*_i(v_i,{u}_{-i})\right\}_{i\in\mathcal{N}})$ as a fixed point
    \begin{equation*}
        \begin{aligned}
        &\mathcal{L}({u}^*;{u}^*;\left\{\zeta^*_i(u^*_i,{u}^*_{-i})\right\}_{i\in\mathcal{N}})\\
        &\qquad\qquad\qquad\quad\leq \mathcal{L}(v;{u}^*;\left\{\zeta^*_i(u^*_i,{u}^*_{-i})\right\}_{i\in\mathcal{N}}).
    \end{aligned}
    \end{equation*}
    That is, $\tilde{f}({u}^*;{u}^*)\leq  \tilde{f}(v;{u}^*),\ \forall v\in \Omega({u}^*).$ By Lemma~\ref{suffice}, the theorem follows.
\end{proof}

\subsection{Dual Game}\label{dualgame}
Under the regularity condition, we can convert the constrained game defined by $\Xi$ to two unconstrained games. To this end, it is sufficient to show that \eqref{mostimportant} is equivalent to the following solution that involves four steps.
\begin{enumerate}
    \item[(i)] Let $T_i({u}_{-i},{\mu}_i)$, $i\in\mathcal{N}$, be the minimizing solution to the following Player $i$'s optimization problem for a fixed $\mu_i$ and $u_{-i}$:
   \begin{equation}\label{mini3}
    \min_{v_i\in\mathcal{U}_i}\ \mathcal{L}_i\left(v_i;{u}_{-i};{\mu}_i\right),
    \quad i\in\mathcal{N}.
    \end{equation}
    \item[(ii)] For a fixed $\{\mu_i\}_{i\in\mathcal{N}}$, compute the fixed points generated by the collection of mappings $\{T_i\}_{i\in\mathcal{N}}$ that maps from $\prod_{i=1}^N\mathcal{U}_i$ to itself.
    For a given $\mu=(\mu_1, \cdots,\mu_N) \in \prod_{i=1}^N\mathbb{R}_+^{|\tilde{\mathcal{M}}_i|}$, denote the fixed points by
    ${V}({\mu}):=({V}_i({\mu}))_{i\in\mathcal{N}}$, where $V_i: \prod_{i=1}^N\mathbb{R}_+^{|\tilde{\mathcal{M}}_i|} \rightarrow \prod_{j\neq i, j\in\mathcal{N}}\mathcal{U}_j$.
    \item[(iii)]  For a given $u_{-i}$, let $K_i({u}_{-i})$, $i\in\mathcal{N}$, be the maximizing solution to the following problem of Player $i$. 
      \begin{equation}\label{max3}
         \max_{{\mu}_i\geq 0}\min_{v_i\in\mathcal{U}_i}\ \mathcal{L}_i\left(v_i;{u}_{-i};{\mu}_i\right),
    \quad i\in\mathcal{N}.
    \end{equation}
   Here $K_i:\prod_{j\neq i, j\in\mathcal{N}}\mathcal{U}_j\rightarrow \mathbb{R}_+^{|\tilde{\mathcal{M}}_i|}$ is a map from the actions of the other players to the Lagrangian multiplier associated with Player $i$. 
    \item[(iv)] Compute the fixed points generated by the collection of the composed mappings $\{K_i \circ V_i\}_{i\in\mathcal{N}}$ that maps from $\prod_{i=1}^N\mathbb{R}_+^{|\tilde{\mathcal{M}}_i|}$ to itself.
    Denote the resulting fixed points  by ${\mu}^*=({\mu}_i^*)_{i\in\mathcal{N}}$. Then, the equilibrium solution can be obtained by $    u^*=({u}_i^*)_{i\in\mathcal{N}}=({V}_i({\mu}^*))_{i\in\mathcal{N}}.$
\end{enumerate}

Note that each of the fixed-point equations in (ii) and (iv) corresponds to an unconstrained game, respectively. In the first unconstrained game, with fixed Lagrangian multipliers, the players play an unconstrained game using the same action as in $\Xi$. In the second unconstrained game, the decision variables of the players are the Lagrangian multipliers.

\begin{theorem}
Under Assumption 1, \eqref{mostimportant} and (i) - (iv) generate the same fixed points.
\end{theorem}

\begin{proof}
On closer inspection, we notice that \eqref{mostimportant} can be decomposed into a maximization, a minimization and a set of fixed point equations. Under Assumption~\ref{bigassump}, as there is no duality gap in the minimax problem \cite{boyd2004convex}, we can interchange the minimization and the maximization as follows:
\begin{enumerate}
        \item \textbf{Minimization}
        \begin{equation}\label{step1}
         \min_{v_i}\ \mathcal{L}_i\left(v_i;{u}_{-i};{\mu}_i\right),
    \quad i\in\mathcal{N};
    \end{equation}
    \item \textbf{Maximization}
        \begin{equation}\label{step2}
         \max_{{\mu}_i}\min_{v_i}\ \mathcal{L}_i\left(v_i;{u}_{-i};{\mu}_i\right),
    \quad i\in\mathcal{N};
    \end{equation}
    \item \textbf{Fixed Point Equation}$\ $\\ Solve a set of fixed point equations:
    \begin{equation}\label{fixedpointeq}
         T_i(K_i({u}_{-i}),{u}_{-i})=u_i\quad, i\in\mathcal{N}.
    \end{equation}
\end{enumerate}

For any ${u}^*$ obtained from (i) - (iv), we only need to show that
\begin{equation*}
    T_i(K_i({u}_{-i}^*),{u}_{-i}^*)=u_i^*,\quad i\in\mathcal{N}.
\end{equation*}
From (ii), we know that  
\begin{equation}\label{dualineq1}
K_i({u}_{-i}^*)={\mu}_i^*,\quad i\in\mathcal{N}.
\end{equation}
And from (iv),
\begin{equation}\label{dualineq2}
    T_i({u}_{-i}^*,{\mu}^*_i)=u_i^*,
\end{equation}
By \eqref{dualineq1} and \eqref{dualineq2},
\begin{equation*}
\begin{aligned}
    T_i({u}_{-i}^*,{\mu}^*_i)= T_i({u}_{-i}^*,K_i({u}_{-i}^*))=u_i^*,\quad i\in\mathcal{N}.
\end{aligned}
\end{equation*}

To prove the other direction. Given ${u}^*$ satisfying,
\begin{equation*}
        T_i({u}_{-i}^*,K_i({u}_{-i}^*))=u_i^*,\quad i\in\mathcal{N},
\end{equation*}
we need to show that it also satisfies {Step 4$^\prime$}. Let
\begin{equation*}
    {\mu}^*_i=K_i({u}_{-i}^*),\quad i\in\mathcal{N}.
\end{equation*}
Then apparently, 
    \begin{equation*}
    T_i({u}^*_{-i},{\mu}^*_i)=u^*_i,\ i\in\mathcal{N}
    \quad \Rightarrow\quad
    {V}({\mu}^*)=({V}_i({\mu}^*))_{i\in\mathcal{N}}.
    \end{equation*}
And 
\begin{equation*}
    K_i\left({V}_{-i}({\mu}^*)\right)={\mu}_i^*\quad i\in\mathcal{N},
\end{equation*}
where ${V}_{-i}({\mu}^*)=({V}_{j}({\mu}^*))_{j\in\mathcal{N}\setminus\{i\}}$. This completes the proof.
\end{proof}

\begin{remark}
The equivalence between \eqref{mostimportant} and (i) - (iv) has an important implication.  We can reformulate a locally-aware constrained game as a two-layer unconstrained game. This result is in contrast to the results in \cite{pavel2007extension} and \cite{bacsar1998dynamic}, where the authors have shown that a globally-aware constrained game can be decomposed into a lower-level modified Nash game with no coupled constraints and a higher-level optimization problem.
\end{remark}

\begin{remark}
In the game induced by the fixed-point equation in (iv), the action is the awareness variable (Lagrange multiplier). After solving for the equilibrium, we can separate the constraints into two classes: inactive constraints and active constraints. If Player $i$ is more aware and has an additional Lagrange multiplier corresponding to an inactive constraint, then the equilibrium set stays the same as the inactive constraint do not affect the equilibria. If Player $i$ has an additional Lagrange multiplier corresponding to an active constraints $\hat{\Omega}_m$, then
\begin{equation*}
    \left\{\left(BR_i(u_{-i}),u_{-i}\right)\right\}\bigcap \hat{\Omega}_m\subseteq \left\{\left(\widetilde{BR}_i(u_{-i}),u_{-i}\right)\right\}\bigcap \hat{\Omega}_m.
\end{equation*}
where $\widetilde{BR}_i$ is the best response with an additional Lagrange multiplier. In a nutshell, solving for GNEs in the new game formed by introducing an additional Lagrange multiplier using Theorem~\ref{neceandsuffic} is a necessary condition for the original game. This result coincides with the result of Theorem~\ref{awarenessandconnectivity}.
\end{remark}

To see the relationship between this result and VEs in the globally-aware constrained game, we consider the following case, where each player has the same constraints, i.e.,
$g(u_i,u_{-i}):=g_i(u_i,u_{-i}),\ {i\in\mathcal{N}}.$
If we impose the following condition $
    \mu_i\ =\ \mu_{j},\ \forall i,j\in\mathcal{N},$
in addition to (i) - (iv), we notice that they give rise to a necessary condition to Theorem 3 in \cite{pavel2007extension}, which has characterized the VEs in a globally-aware constrained game. 
\begin{remark}
In \cite{pavel2007extension}, the players share the same Lagrange multipliers. The interpretation of players sharing the same  Lagrange multipliers is the following. If the players share the same Lagrange multipliers, then they have the global awareness about the constraint and this awareness is common knowledge. Using the common knowledge of the global awareness, the players know that they could potentially select a better equilibrium from GNEs. They agree to pose an additional variational inequality leading to the VE.
\end{remark}

\section{Case Studies}\label{Application}
In this section, we present two case studies to corroborate the  analytical results.
\subsection{Cournot Game and Bertrand Game}
The Cournot Games \cite{tirole1988theory}, where players do not have access to the other player's constraints, are a special case of the locally-aware constrained games. Let $u_i$ denote the production level of firm $i$, $i\in\mathcal{N}=\{1,2\}$. Denote the price of the commodity by $p$. In addition, we assume a linear structure of the market demand curve,
$p = \alpha - \beta\sum_{i\in\mathcal{N}}u_i$, where $\alpha  \in \mathbb{R}_+$ and $\beta \in \mathbb{R}_+$ are publicly known positive constants.

Assume that both firms have a quadratic production cost and different market demands. Then, the Cournot game between two firms can be formulated as follows: for each $i\in\mathcal{N}$,
\begin{align*}
\text{Firm $i$}: \quad\max_{u_i\geq 0}&\quad -k_i u_i^2 +p u_i \\
\text{s.t.} &\quad u_1+u_2\geq q_i,
\end{align*}
 where $q_i \in \mathbb{R}_+$ and $k_i  \in \mathbb{R}_+$ are positive constants. 
The Bertrand games \cite{tirole1988theory} are a class of games where firms choose prices to compete in a market. We denote by $p_1 \in \mathbb{R}_+$ and $p_2 \in \mathbb{R}_+$ the price of product $1$ and product $2$, respectively. The demands for firm $1$ and firm $2$ are 
$
u_1 = \delta - \eta(p_1 - p_2),\
u_2 = \delta - \eta(p_2 - p_1),
$ where $\delta$ and $\eta$ are positive constants.
The cost function for firm $i$, $i\in\mathcal{N}=\{1,2\}$, is
\begin{equation*}
    \begin{aligned}
        d_i(p_i,p_{-i}) = c_i\left( \delta - \eta(p_i - p_{-i}) \right)^2 - \left(\delta - \eta(p_i - p_{-i})\right)p_{i},
\end{aligned}
\end{equation*}
where $c_i$ is the marginal cost. Then, the Bertrand game can be formulated as follows: for $i\in\mathcal{N}$,
\begin{equation}\label{bertrand}
    \begin{aligned}
       \text{Firm }i: \quad \min_{p_i\geq 0}\ d_i(p_i,p_{-i}).
\end{aligned}
\end{equation}
Note that BG is an unconstrained quadratic pricing game. 

\subsubsection{Dual Cournot Game (DCG)}
Following Step (i) - (iv), we observe that the fixed-point equation in (iv) corresponds to the following unconstrained game: for $i\in\mathcal{N}$,
\begin{equation}\label{dual_cournot}
    \text{Firm }i: \quad \min_{\mu_i\geq0}\quad \frac{1}{2}A_i\mu_i^2 -B_i\mu_1\mu_2-C_i\mu_i 
+ F_i(\mu_{-i}),
\end{equation}
where
\begin{equation*}
\begin{aligned}
    A_i&= 1+\frac{\beta(2k_1+\beta)}{H}, \\
    B_i&=\frac{2(k_1+\beta)(2k_1+\beta)}{H},\\
    C_i&=\frac{2\alpha(k_1+\beta)(2k_1+\beta)}{H}  + {2(k_1+\beta)}q_1-\alpha,\\
    H&=4\prod_{i=1}^2(k_i+\beta)-\beta^2,
\end{aligned}
\end{equation*}
and $F(\mu_{-i})$ is a function which does not influence the optimization result or the equilibrium. 

\subsubsection{Equivalence Between DCG and BG}
For given $\alpha$, $\beta$, and $k_i$, $i\in\mathcal{N}$, in the Cournot game, and by matching the coefficients in \eqref{dual_cournot} and \eqref{bertrand} in the following way: for $i\in\mathcal{N}$,
\begin{equation}\label{matching}
\begin{aligned}
    \eta+c_i\eta^2&=\frac{1}{2}\left( 1+\frac{\beta(2k_i+\beta)}{H} \right),  \\
\eta + 2c_i \eta^2&=\frac{2(k_i+\beta)(2k_i+\beta)}{H}, \\
\delta + 2c_i \delta\eta&=\frac{2\alpha(k_i+\beta)(2k_i+\beta)}{H}  + {2(k_i+\beta)}q_i-\alpha .
\end{aligned}
\end{equation}
So, we can obtain parameters $\delta$, $\eta$, $c_1$, and $c_2$ for the Bertrand game. After algebraic simplification, we obtain that for $i\in\mathcal{N}$,
\begin{equation*}
\begin{aligned}
    \eta = 1 - \frac{(2k_i+\beta)^2}{H},
    \end{aligned}
\end{equation*}
and
\begin{equation*}
\begin{aligned}
    c_i = \frac{H\left((2k_i+\beta)(4k_i+3\beta)-H\right)}{2\left(H-(2k_i+\beta)^2\right)^2}.
\end{aligned}
\end{equation*}
Therefore, solving DCG \eqref{dual_cournot} is equivalent to solving the Bertrand game \eqref{bertrand} with parameters specified in \eqref{matching}, and we have $\mu_i^* = p_i^*$, $i\in\mathcal{N}$. The advantage of this equivalence lies in the fact that solving one game automatically gives the results of the other game. In addition, the equivalent dual game (BG) can be easier to solve comparing with the primal game (CG).


\subsection{Linear-Quadratic Games}
\subsubsection{Unique GNE Case}
Consider a linear-quadratic game with linear constraints. For each player $i\in\mathcal{N}=\{1,2\}$, the cost function is $f_i(u_i,u_{-i})=\alpha u_i-\frac{1}{2}u^2_i+\phi u_{-i}u_i$, $u_i\in\mathcal{U}_i=\mathbb{R}$, and the constraints are $\hat{\Omega}_1=\{(u_1,u_2):u_1-c_1u_{2}-d_1\leq 0\}$ and  $\hat{\Omega}_2=\{(u_1,u_2):u_2-c_2u_{1}-d_2\leq 0\}$. 
We assume that $\alpha> 0$ and $\phi> 1$. Furthermore, we assume that $c_i> \phi$, $i\in\mathcal{N}$. As shown in Fig. \ref{UniqueGNE}, when there is only one GNE in $\tilde{\Xi}$, then under different awareness level structures, the corresponding ${\Xi}$ always admits the same GNE as $\tilde{\Xi}$.

\begin{figure}[!tp]
  \includegraphics[clip,width=1\columnwidth]{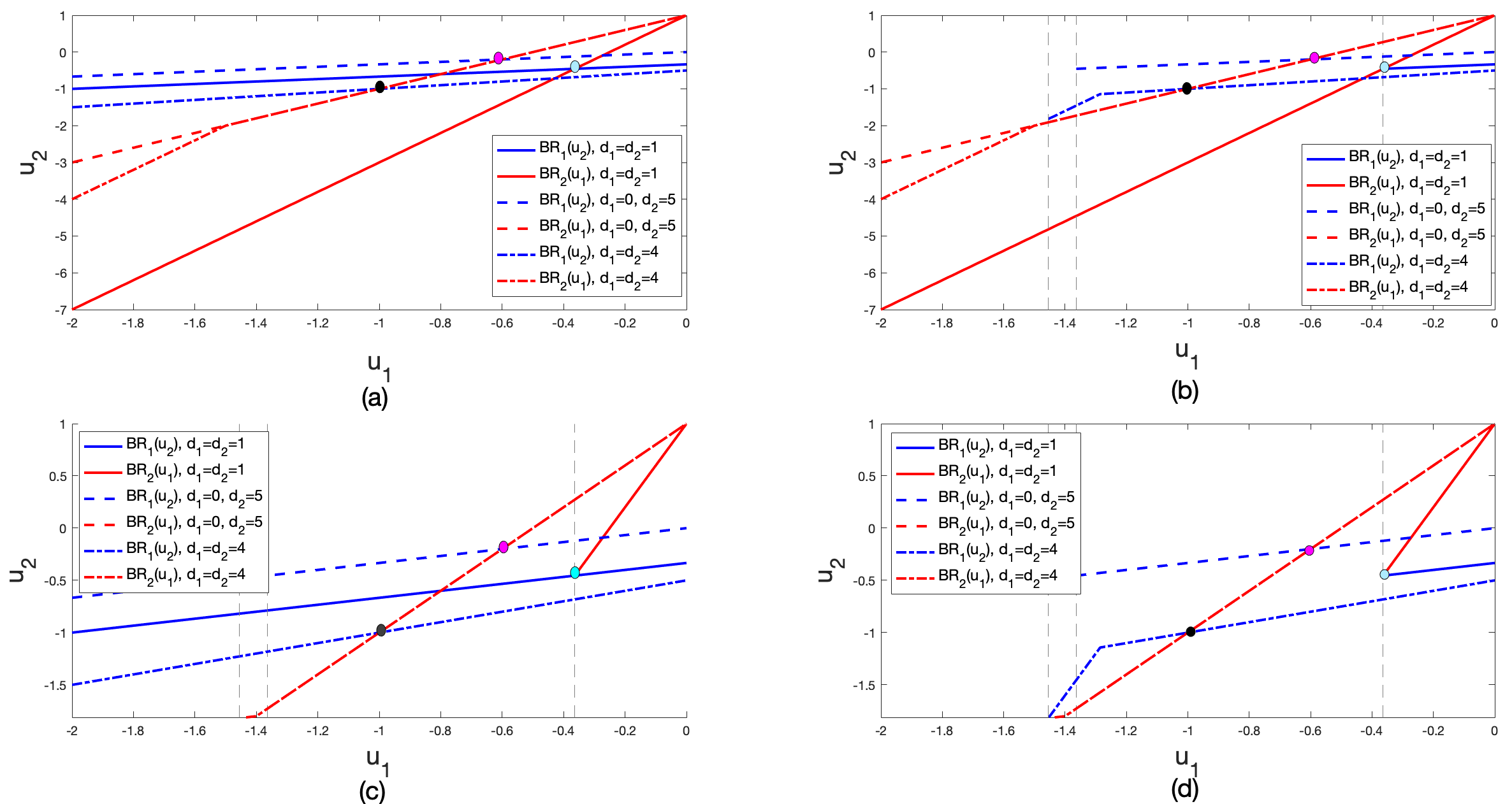}
\caption{An example of unique GNE in $\tilde{\Xi}$, where $\alpha=1$, $\phi=2$, $c_1=3$, and $c_2=4$: (a) depicts the case where Player $1$ is aware of $\hat{\Omega}_1$, and Player $2$ is aware of $\hat{\Omega}_2$; (b) depicts the case where Player $1$ is aware of $\hat{\Omega}_1$ and $\hat{\Omega}_2$, and Player $2$ is aware of $\hat{\Omega}_2$; (c) depicts the case where Player $1$ is aware of $\hat{\Omega}_1$, and Player $2$ is aware of $\hat{\Omega}_1$ and $\hat{\Omega}_2$; (d) depicts the case where both Player $1$ and Player $2$ are aware of $\hat{\Omega}_1$ and $\hat{\Omega}_2$. }\label{UniqueGNE}
\end{figure}

\subsubsection{Non-Unique GNE Case}

When these are multiple GNE solutions to $\tilde{\Xi}$, consider the following case: $f_1(u_1,u_2)=u_1^2+2u_1u_2+2u_2^2$, $f_2(u_1,u_2)=\ u_1^2+2u_1u_2+2u_2^2$, $\mathcal{U}_1=\mathcal{U}_2=\mathbb{R}$, $\hat{\Omega}_1=\{(u_1,u_2):u_1+u_2\leq 3\}$, and $\hat{\Omega}_2=\{(u_1,u_2):2u_1+u_2\leq -5\}$.
As shown in Fig. \ref{NonUniqueGNE}, as the awareness levels of the players increase, the set of GNEs becomes larger. Moreover, the VE is $(-2,-1)$, which coincides with one of the GNE. This is because when $(u^*_1,u^*_2)=(-2,-1)$, the corresponding Lagrange multipliers $\mu^*_1=\mu^*_2=2\geq 0$.

\begin{figure}[htp]
  \includegraphics[clip,width=1\columnwidth]{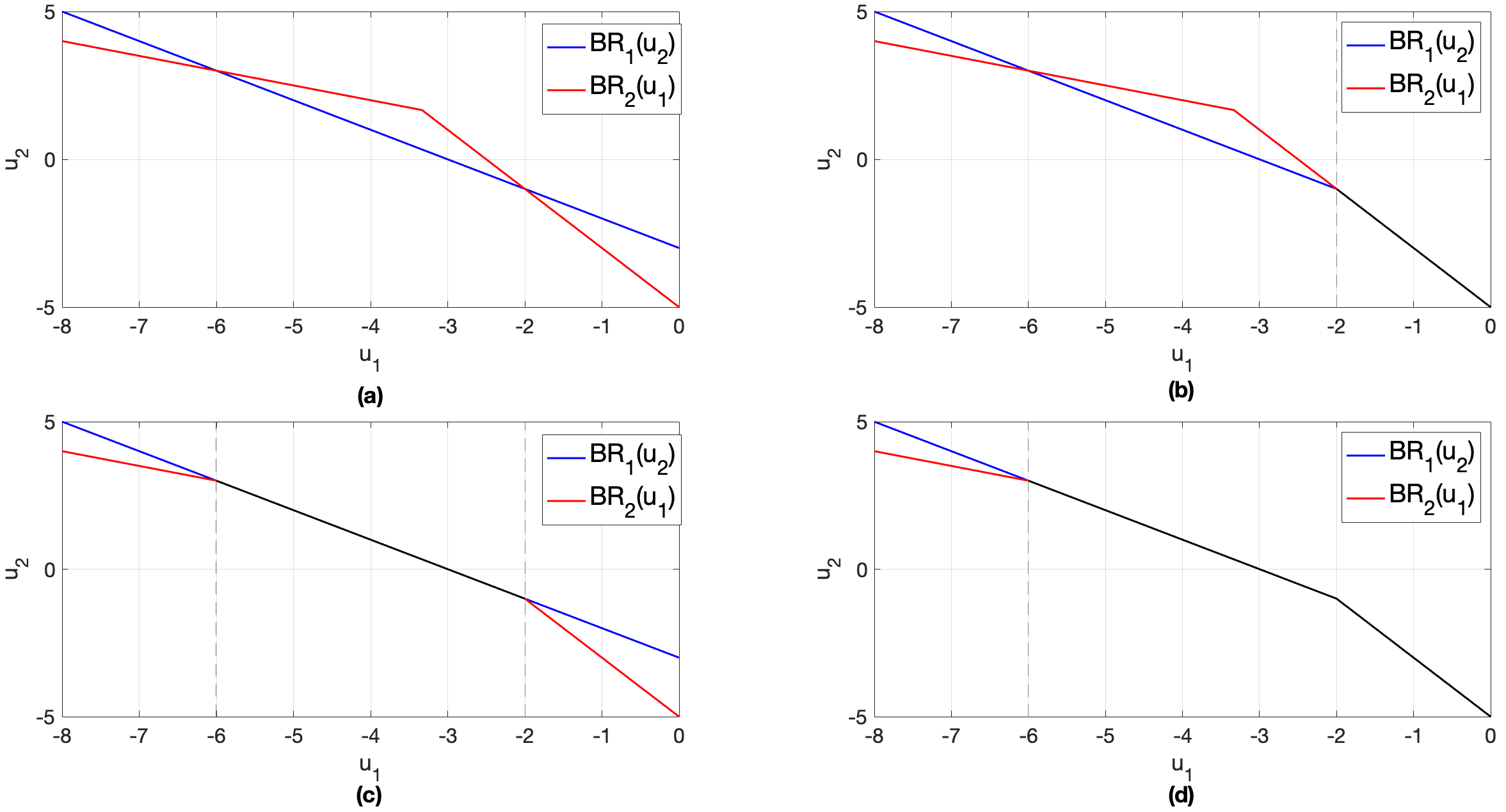}
\caption{An example of multiple GNEs in $\tilde{\Xi}$. The intersection of best response curves or the black line stands for the GNEs. The players' awareness levels in each case are the same as the corresponding ones described in Fig. 1.
}\label{NonUniqueGNE}
\end{figure}

 {
\subsubsection{An Application to Energy System} We next apply the developed framework to an interconnected distributed generation scenario. Specifically, we consider two players in the game and each determines the amount of power to generate. For each player $i\in\mathcal{N}=\{1,2\}$, the generation cost function is $f_i(u_i)= \alpha_i u_i+\beta_i u^2_i$, $u_i\in\mathcal{U}_i=\mathbb{R}_+$, and the constraints are $\hat{\Omega}_1=\{(u_1,u_2):u_1+u_2\leq p_1\}$ and  $\hat{\Omega}_2=\{(u_1,u_2):u_1+u_2\geq p_2\}$, where $p_1\geq p_2$. Note that the constraints $\hat{\Omega}_1$ and $\hat{\Omega}_2$ can be interpreted as the capacity of the two connected generation systems and the load to satisfy, respectively. The results under various awareness structures are shown in Fig. \ref{powersysm}, where the parameters are chosen as $\alpha_1=\alpha_2=0\$/\mathrm{MW}$, $\beta_1=\beta_2=1\$/\mathrm{MW}^2$, $p_1=2$MW, and $p_2=1$MW. An observation is that when player 2 is only aware of the demand constraint $\hat{\Omega}_2$ (Fig. \ref{powersysm}(a) and Fig. \ref{powersysm}(c)), then the game has a unique GNE at which player 1 does not generate energy and player 2's generation is 1MW. In comparison, when player 2 is aware of both constraints $\hat{\Omega}_1$ and $\hat{\Omega}_2$, the game admits multiple GNE (Fig. \ref{powersysm}(b) and Fig. \ref{powersysm}(d)) at which both players generate power with a total 1MW to meet the demand. By comparing Fig. \ref{powersysm}(a) with  Fig. \ref{powersysm}(b) (Fig. \ref{powersysm}(c) with Fig. \ref{powersysm}(d)), it again confirms that the set of GNEs becomes larger as player's awareness levels increases.}

\begin{figure}[htp]
  \includegraphics[clip,width=1\columnwidth]{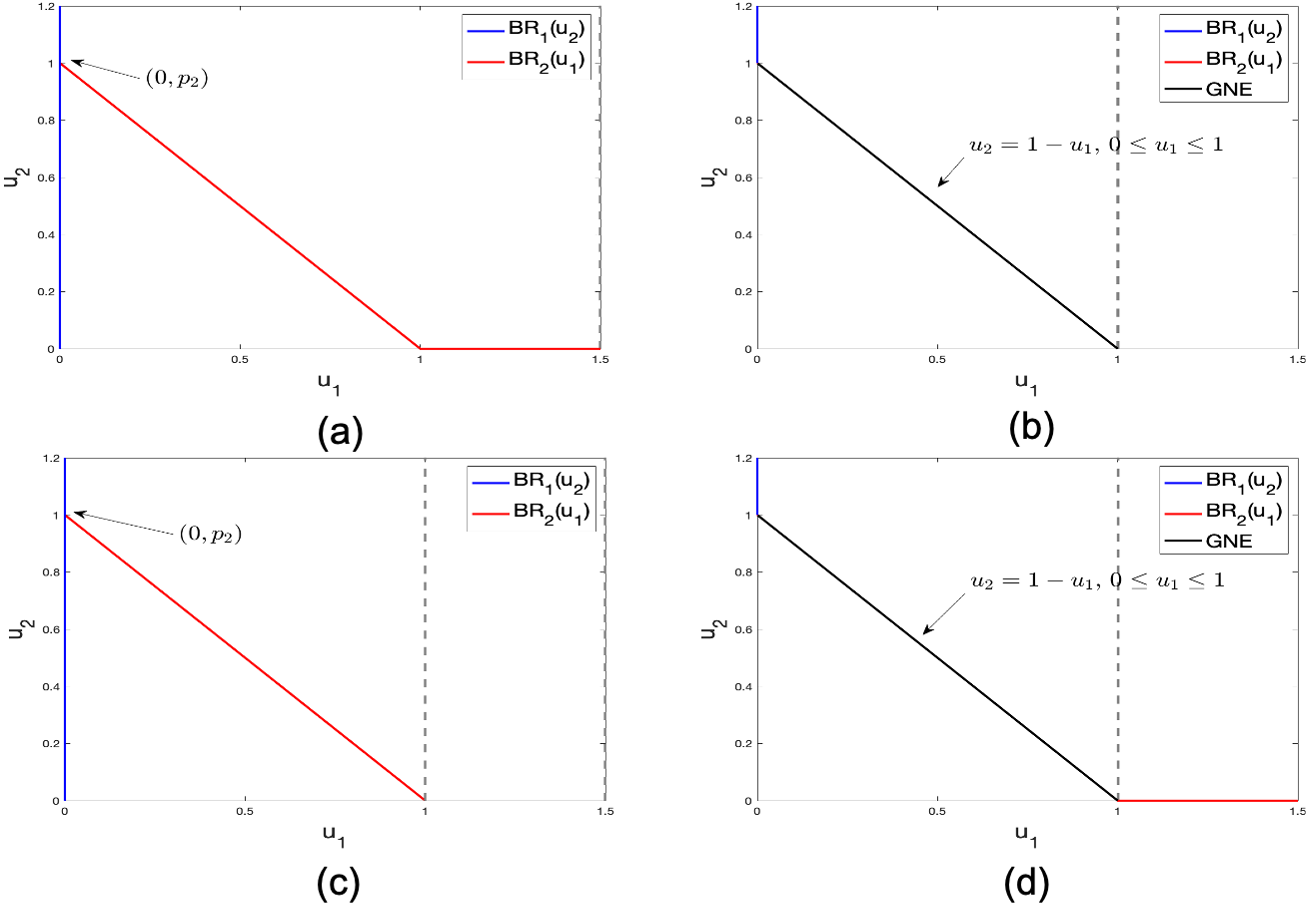}
\caption{ {Various GNEs under different awareness level of players in $\tilde{\Xi}$. The intersection of best response curves or the black line stands for the GNEs. The players' awareness levels in each case (a)-(d) are the same as the corresponding ones described in Fig. 3.  
} 
}\label{powersysm}
\end{figure}

\section{Conclusions and Future Work}\label{CONCLUSIONS AND FUTURE WORK}

In this work, we have studied a class of constrained games on networks where the players have asymmetric information regarding the constraints. We have proposed a new concept, \textit{awareness}, to quantify the players' knowledge about the constraints. We have shown that the growth in network connectivity is equivalent to the increase in the awareness level. When there is a unique GNE in a globally-aware game, then the GNE in the corresponding locally-aware game is also unique, and these two GNEs coincide. When the connectivity of the network increases, the game admits more GNEs. We have illustrated the relationship between the GNE and the VE in the globally-aware games using the Lagrangian approach: to find the VE using Theorem~\ref{neceandsuffic}, we need to impose an additional constraint on the Lagrange multipliers, which requires the Lagrange multiplier vector of each player to be the same. We have shown that the constrained game considered in this work can be decomposed into two unconstrained games. Using this decomposition, we have shown that the Cournot game is equivalent to a Bertrand game.

As for future work, we would investigate the design of network leveraging the awareness of the constraints to achieve desirable network performances from the designer's perspective. 
In addition, we have assumed that the awareness levels of the players are chosen by nature. We could also explore the case where the awareness level can be determined by the players, i.e., the players can choose to ignore or be aware of the constraints.

\addtolength{\textheight}{-12cm}   








\bibliographystyle{IEEEtran}
\bibliography{reference.bib}

\end{document}